\documentclass[sts]{imsart}

\usepackage[margin=0.95in]{geometry} 
\usepackage[english]{babel}
\usepackage{graphicx}
\usepackage{caption}
\usepackage{amsthm,amsmath,amsfonts,bm,amssymb}
\usepackage{natbib}
\usepackage{array,dcolumn,booktabs,float}
\RequirePackage[colorlinks,citecolor=blue,urlcolor=blue]{hyperref}
\usepackage{color}
\usepackage[ruled,norelsize,vlined]{algorithm2e}
\usepackage{etoolbox}
\AtBeginEnvironment{algorithm}{\linespread{1.5}\selectfont}

\startlocaldefs

\newcommand{\mockalph}[1]{}

\newtheorem{theorem}{Theorem}

\newtheorem{lemma}{Lemma}

\newtheorem{proposition}{Proposition}

\newcommand{\bSigma}{ {\boldsymbol \Sigma} }
\newcommand{\btheta}{ {\boldsymbol \theta} }
\newcommand{\bTheta}{ {\boldsymbol \Theta} }
\newcommand{\bmu}{ {\boldsymbol \mu} }

\newcommand{\boeta}{ {\boldsymbol \eta} }
\newcommand{\bbeta}{ {\boldsymbol \beta} }
\newcommand{\blambda}{ {\boldsymbol \lambda} }
\newcommand{\bx}{ {\bf x} }
\newcommand{\bX}{ {\bf X} }
\newcommand{\bZ}{ {\bf Z} }
\newcommand{\bJ}{ \bm{J} }
\newcommand{\bLambda}{ {\boldsymbol \Lambda} }

\newcommand{\bGamma}{ {\boldsymbol \Gamma} }
\newcommand{\by}{ {\bf y} }
\newcommand{\bH}{ {\bf H} }
\newcommand{\bz}{ {\bf z} }
\newcommand{\bxi}{ {\boldsymbol \xi} }

\newcommand{\bpi}{ {\boldsymbol \pi} }

\def\T{\intercal}
\endlocaldefs

\begin{document}

\begin{frontmatter}

\title{Conditionally conjugate mean--field variational Bayes for logistic models}
\runtitle{Conditionally conjugate mean--field variational Bayes for logistic models}

\begin{aug}
\author{\fnms{Daniele} \snm{Durante}\ead[label=e1]{daniele.durante@unibocconi.it}}
\and
\author{\fnms{Tommaso} \snm{Rigon}
\ead[label=e2]{tommaso.rigon@phd.unibocconi.it}}

\runauthor{Durante and Rigon}


\thankstext{}{Department of Decision Sciences and Bocconi Institute for Data Science and Analytics, Bocconi University, Via Roentgen 1, Milan, Italy \printead{e1,e2}.}
\end{aug}

\begin{abstract}
Variational Bayes (\textsc{vb}) is a common strategy for approximate Bayesian inference, but simple methods are only available for specific classes of models including, in particular,  representations having conditionally conjugate constructions within an exponential family. Models with logit components are an apparently notable exception to this class, due to the absence of conjugacy between the logistic likelihood and the Gaussian priors for the coefficients in the linear predictor. To facilitate approximate inference within this widely used class of models, \citet{jak_2000} proposed a simple variational approach which relies on a family of tangent quadratic lower bounds of  logistic log-likelihoods, thus restoring conjugacy between these approximate bounds and the Gaussian priors. This strategy is still  implemented successfully, but less attempts have been made to formally understand the reasons underlying its excellent performance. To cover this key gap, we provide a formal connection between the above bound and a recent P\'olya-gamma data augmentation for logistic regression. Such a result places the computational methods associated with the aforementioned bounds within the framework of variational inference for conditionally conjugate exponential family models, thereby allowing recent advances for this class to be inherited also by the methods relying on   \citet{jak_2000}.
\end{abstract}



\end{frontmatter}

\section{Introduction}
\label{sec1}
The increasing availability of massive and high--dimensional datasets has motivated a wide interest in strategies for Bayesian learning of posterior distributions, beyond classical \textsc{mcmc} methods \citep[e.g.][]{gelfand_1990}. Indeed, sampling algorithms can face severe computational bottlenecks in complex statistical models, thus motivating alternative solutions based on scalable and efficient optimization of approximate posterior distributions. Notable methods within this class are the Laplace approximation \citep[e.g.][Ch. 4.4]{bishop_2006}, variational Bayes  \citep[e.g.][Ch. 10.1]{bishop_2006} and expectation propagation  \citep[e.g.][Ch. 10.7]{bishop_2006}, with variational inference providing a standard choice in several fields, as discussed in  recent reviews by  \citet{blei_2017} and \citet{ormerod2010}. Refer also to \citet{jordan_1999} for a seminal introduction of variational inference from a statistical perspective.

Adapting the notation in  \citet{blei_2017}, \textsc{vb} aims at obtaining a tractable approximation $q^*(\btheta)$ for the posterior distribution $p(\btheta \mid \by)$ of the  random coefficients $\btheta=(\theta_1, \ldots, \theta_m)^{\intercal}$, in the model having joint density $p(\by, \btheta)=p(\by \mid \btheta)p(\btheta)$ for $\btheta$ and the observed data $\by=(y_1, \ldots, y_n)^{\intercal}$, with $p(\btheta)$ denoting the prior distribution for $\btheta$. This optimization problem is formally addressed by minimizing  the Kullback--Leibler  (\textsc{kl}) divergence \citep{kullback_1951}
\begin{eqnarray}
\textsc{kl}[q(\btheta)\mid\mid p(\btheta \mid \by)]= \int_{\bTheta}q(\btheta)\log \frac{q(\btheta)}{p(\btheta \mid \by)} \mbox{d}\btheta= \int_{\bTheta}q(\btheta)\log \frac{q(\btheta)p(\by)}{p(\by,\btheta)} \mbox{d}\btheta,
\label{eq1}
\end{eqnarray}
with respect to $q(\btheta) \in \mathcal{Q}$, where $\mathcal{Q}$ denotes a tractable, yet sufficiently flexible, class of approximating distributions. As is clear from \eqref{eq1}, the calculation of the $\textsc{kl}$ divergence between $q(\btheta)$ and the posterior $p(\btheta \mid \by)$ requires the evaluation of the normalizing constant $p(\by)$, whose intractability is actually the main reason motivating approximate Bayesian methods. Due to this, the above minimization problem is commonly translated into the maximization of the evidence lower bound (\textsc{elbo}) function
\begin{eqnarray}
\textsc{elbo}[q(\btheta)]= \int_{\bTheta}q(\btheta)\log \frac{p(\by,\btheta)}{q(\btheta)} \mbox{d}\btheta=-\textsc{kl}[q(\btheta)\mid\mid p(\btheta \mid \by)]+ \log p(\by),
\label{eq2}
\end{eqnarray}
which does not require the evaluation of $p(\by)$. In fact, since $\log p(\by)$ does not depend on $\btheta$, maximizing \eqref{eq2} is equivalent to minimizing \eqref{eq1}. Re--writing \eqref{eq2} as $ \log p(\by)=\textsc{elbo}[q(\btheta)]+\textsc{kl}[q(\btheta)\mid\mid p(\btheta \mid \by)]$ it can be additionally noticed that the \textsc{elbo} provides a lower bound of $ \log p(\by)$ for any $q(\btheta)$, since the Kullback--Leibler divergence is always non--negative \citep{kullback_1951}.

The above set--up defines the general rationale underlying \textsc{vb} but, as is clear from \eqref{eq2}, the practical feasibility of the variational optimization requires a tractable form for the joint density $p(\by,\btheta)$ along with a simple, yet flexible, variational family $\mathcal{Q}$. This is the case of mean--field  \textsc{vb} for conditionally conjugate exponential family models with global and local variables \citep{wang2004, bishop_2006, hoffman_2013, blei_2017}. Recalling \citet{hoffman_2013}, these methods focus on obtaining a mean--field approximation 
\begin{eqnarray}
\begin{split}
q^*(\btheta)=q^*(\bbeta, \bz)=q^*(\bbeta) \prod\nolimits_{i=1}^nq^*(z_i)=\mbox{argmin}\{\textsc{kl}[q(\bbeta) \prod\nolimits_{i=1}^nq(z_i) \mid \mid p(\bbeta, \bz \mid \by)]\},&\\
=\mbox{argmax}\{\textsc{elbo}[q(\bbeta) \prod\nolimits_{i=1}^nq(z_i)]\}\qquad \qquad  \ &
\label{eq31}
\end{split}
\end{eqnarray}
for the posterior distribution $p(\bbeta, \bz \mid \by)$ of the global coefficients $\bbeta=(\beta_1, \ldots, \beta_p)^{\intercal}$ and the local variables $\bz=(z_1, \ldots, z_n)^{\intercal}$ in the statistical model having joint density
\begin{eqnarray}
p(\by,\bbeta, \bz)=p(\bbeta) \prod\nolimits_{i=1}^np(z_i \mid \bbeta)p(y_i \mid z_i, \bbeta)=p(\bbeta) \prod\nolimits_{i=1}^np(y_i, z_i \mid \bbeta),
\label{eq3}
\end{eqnarray}
with $p(y_i, z_i \mid \bbeta)$ from an exponential family and $p(\bbeta)$ being a conjugate prior for this density. The latent quantities $\bz$---when present---typically denote random effects or unit--specific augmented data within some hierarchical formulation, such as in mixture models. 

Although the above assumptions appear restrictive, the factorization of $q(\bbeta,\bz)$---characterizing the mean--field variational family---provides a flexible class in several applications and allows direct implementation of simple coordinate ascent variational inference (\textsc{cavi}) routines \citep[Ch. 10.1.1]{bishop_2006} which sequentially maximize the \textsc{elbo} in \eqref{eq31}  with respect to each factor in $q(\bbeta, \bz)=q(\bbeta) \prod_{i=1}^nq(z_i)$---fixing the others at their most recent update. Instead, the exponential family and conjugacy assumptions further simply calculations by providing approximating densities $q^*(\bbeta)$ and $q^*(z_i)$, $i=1, \ldots, n$ from tractable classes of random variables. These advantages have also motivated recent computational improvements \citep{hoffman_2013} and theoretical studies \citep{wang2004}.  We refer to \citet{hoffman_2013} and \citet{blei_2017} for details on the methods related to the general formulation in \eqref{eq31}--\eqref{eq3}, and focus here on models having logistic likelihoods as building--blocks. Indeed, although the conjugacy and exponential family assumptions are common to a variety of machine learning representations \citep[e.g.][]{blei_2003,airoldi_2008, hoffman_2013}, classical Bayesian logistic regression models of the form
\begin{eqnarray}
p(y_i \mid \bbeta)=\frac{[\exp(\bx_i^{\intercal} \bbeta)]^{y_i}}{1+\exp(\bx_i^{\intercal} \bbeta)}, \ \ y_i \in \{0,1\}, \ i=1, \ldots, n, \quad  \mbox{with } \bbeta \sim \mbox{N}_p(\bmu_0, \bSigma_0),
\label{eq4}
\end{eqnarray}
do not enjoy direct conjugacy between the likelihood for the binary response data and the Gaussian prior for the coefficients in the linear predictor \citep[e.g.][]{wang_2013}. This apparently notable exception to conditionally conjugate exponential family models also holds, as a direct consequence, for a wide set of formulations which incorporate Bayesian logistic regressions at some layer of the hierarchical specification. Some relevant examples are classification via Gaussian processes \citep{rasm_2006},  supervised nonparametric clustering \citep{Ren2011a} and hierarchical mixture of experts \citep{Bishop2003}.

\begin{algorithm}[t]
\caption{\textsc{em} algorithm for approximate Bayesian inference by \citet{jak_2000}.} \label{al1}
 Initialize $\xi_1^{(0)}, \ldots, \xi^{(0)}_n$.\\
 \For( ){{\normalfont{$t=1$ until convergence of \eqref{eq6}}}}
 {
\enspace   \textsc{\bf E--step}. Update $q^{(t)}(\bbeta)=\bar{p}^{(t-1)}(\bbeta \mid \by)\propto p(\bbeta)\prod_{i=1}^{n}\bar{p}^{(t-1)}(y_i \mid \bbeta)$ to obtain a $\mbox{N}_p(\bmu^{(t)}{,} \bSigma^{(t)})$ density with $$\bSigma^{(t)}=(\bSigma_0^{-1}+\bX^{\intercal} \bar{\bZ}^{(t-1)} \bX)^{-1}, \qquad \bmu^{(t)}=\bSigma^{(t)}[\bX^{\intercal}(\by-0.5{\cdot}{\bf 1}_n)+\bSigma_0^{-1}\bmu_0 ],$$
where ${\bf 1}_n=(1, \ldots, 1)^{\intercal}$ and $ \bar{\bZ}^{(t-1)}=\mbox{diag}\{0.5[\xi_1^{(t-1)}]^{-1}\mbox{tanh}(0.5 \xi^{(t-1)}_1), \ldots,  0.5[\xi_n^{(t-1)}]^{-1}\mbox{tanh}(0.5 \xi^{(t-1)}_n)\}$. Note that the quadratic form of \eqref{eq5}, restores conjugacy between the Gaussian prior for $\bbeta$ and the approximated likelihood. To clarify this result, note that, for every $\xi_i$, $\bar{p}(y_i \mid \bbeta)$ is proportional to the kernel of a normal with mean $\bx_i^{\intercal}\bbeta$ and variance $2\xi_i\mbox{\textup{tanh}}(0.5\xi_i)^{-1}$ for transformed data $2\xi_i\mbox{\textup{tanh}}(0.5\xi_i)^{-1}(y_i-0.5)$.\\
\enspace   \textsc{\bf M--step}. Compute $\bxi^{(t)}=\mbox{argmax}_{\bxi}\int_{\Re^{p}}{q}^{(t)}(\bbeta)\log \bar{p}(\by,\bbeta)\mbox{d}\bbeta$ to obtain the solutions $$\xi^{(t)}_i=\{\mathbb{E}_{{q}^{(t)}(\bbeta)}[(\bx_i^{\intercal} \bbeta)^2 ]\}^{\frac{1}{2}}=[\bx^{\intercal}_i\bSigma^{(t)}\bx_i{+}(\bx_i^{\intercal}\bmu^{(t)})^2]^{\frac{1}{2}}, \quad \mbox{for every } i=1, \ldots, n.$$
Note that $\int_{\Re^{p}}{q}^{(t)}(\bbeta)\log \bar{p}(\by,\bbeta)\mbox{d}\bbeta=\mbox{const}+\sum_{i=1}^n \int_{\Re^{p}}{q}^{(t)}(\bbeta)\log \bar{p}(y_i\mid\bbeta)\mbox{d}\bbeta$. Hence, it is possible to maximize the expected log-likelihood associated with every $y_i$ separately, as a function of each $\xi_i$, for $i=1, \ldots, n$. This result leads to the above solution.}
{\bf Output at the end of the algorithm:} $\bxi^*$ and, as a byproduct, the approximate posterior ${q}^*(\bbeta)=\bar{p}^{*}(\bbeta \mid \by)$.
\end{algorithm}

To allow tractable \textsc{vb} for non--conjugate models, several alternatives beyond conjugate mean--field \textsc{vb} have been proposed \citep[see e.g.][]{jak_2000, braun_2010, wand_2011, wang_2013}. Within the context of logistic regression, \citet{jak_2000} developed a  seminal \textsc{vb} algorithm which relies on the quadratic lower bound 
\begin{eqnarray}
\qquad \log \bar{p}(y_i \mid \bbeta) =(y_i-0.5)\bx_i^{\T}\bbeta-0.5\xi_i-0.25\xi_i^{-1}\mbox{\textup{tanh}}(0.5\xi_i)[(\bx_i^{\T}\bbeta)^2- \xi_i^2]-\log[1+\exp(-\xi_i)], 
\label{eq5}
\end{eqnarray}
for the log-likelihood $\log p(y_i \mid \bbeta)=y_i(\bx_i^{\T}\bbeta)-\log[1+\exp(\bx_i^{\T}\bbeta)] \geq \log \bar{p}(y_i \mid \bbeta)$ of every   $y_i$ from a logistic regression. In \eqref{eq5}, the vector $\bx_i=(x_{i1}, \ldots, x_{ip})^{\T}$ comprises the covariates measured for unit $i$, whereas $\bbeta=(\beta_1, \ldots, \beta_p)^{\T}$ are the associated coefficients. The vector $\bxi=(\xi_1, \ldots, \xi_n)^{\intercal}$ denotes instead unit--specific variational parameters defining the location where $\log \bar{p}(y_i \mid \bbeta)$ is tangent to $\log p(y_i \mid \bbeta)$. In fact, $\log \bar{p}(y_i \mid \bbeta)=\log p(y_i \mid \bbeta)$ when $\xi_i^2=(\bx_i^{\T}\bbeta)^2$. Leveraging equation \eqref{eq5},  \citet{jak_2000}  proposed an expectation--maximization (\textsc{em}) algorithm \citep{demp_1977} to approximate $p(\bbeta \mid \by)$. At the generic iteration $t$, this routine alternates between an \textsc{e}--step in which the conditional distribution  of the random coefficients $\bbeta$ given the current  $\bxi^{(t-1)}$ is updated to obtain ${q}^{(t)}(\bbeta)$, and an \textsc{m}--step which calculates the expectation of the augmented approximate log-likelihood $\log \bar{p}(\by,\bbeta)=\log p(\bbeta)+\sum_{i=1}^n \log \bar{p}(y_i \mid \bbeta)$ with respect to ${q}^{(t)}(\bbeta)$ and maximizes it as a function of $\bxi$. Recalling the general presentation of  \textsc{em} by \citet[][Ch. 9.4]{bishop_2006} and Appendices A--B in \citet{jak_2000}, this strategy ultimately maximizes $\log \bar{p}(\by)=\log \int_{\Re^p}p(\bbeta)\prod_{i=1}^n \bar{p}(y_i \mid \bbeta)\mbox{d} \bbeta$ with respect to $\bxi$, by sequentially optimizing the lower bound 
\begin{eqnarray}
\int_{\Re^p} q(\bbeta)\log \frac{p(\bbeta) \prod_{i=1}^{n} \bar{p}(y_i\mid \bbeta)}{q(\bbeta)} \mbox{d} \bbeta,
\label{eq6}
\end{eqnarray}
as a function of the unknown distribution $q(\bbeta)$ and the fixed parameters $\bxi$, where   $p(\bbeta)$ is the density of the Gaussian prior for $\bbeta$. Hence, as is clear from Algorithm  \ref{al1}, this \textsc{em} produces an optimal estimate $\bxi^{*}$ of $\bxi$ and, as a byproduct, also a distribution ${q}^*(\bbeta)$, which is regarded as an approximate posterior in \citet{jak_2000}. Indeed, recalling the \textsc{em} structure, ${q}^*(\bbeta)$ coincides with the conditional distribution $\bar{p}^{*}(\bbeta \mid \by)$ obtained by updating the prior $p(\bbeta)$ with the approximate likelihood $\prod_{i=1}^n \bar{p}^*(y_i \mid \bbeta)$ induced by \eqref{eq5} and evaluated at the optimal variational parameters $\xi^*_1, \ldots, \xi^*_n$. However, although being successfully implemented  in the machine learning and statistical literature \citep[e.g.][]{Bishop2003, rasm_2006, Lee2010, Ren2011a, Carbonetto2012, Tang2015, Wand2017},  it is not clear how the solution ${q}^*(\bbeta)$ relates to the formal \textsc{vb} set--up in \eqref{eq1}--\eqref{eq2}.  Indeed, $\bar{p}^{*}(\bbeta \mid \by)$ is not the posterior induced by a Bayesian logistic regression. This is due to the fact that each $p(y_i \mid \bbeta)$ in the kernel of $p(\bbeta \mid \by)$ is replaced with the approximate likelihood $\bar{p}^*(y_i \mid \bbeta)$ evaluated at the optimal variational parameters maximizing $\log \bar{p}(\by)$. This last result, which is inherent to the \textsc{em} \citep[][]{demp_1977}, suggests an heuristic intuition  for why ${q}^*(\bbeta)$ may still provide a reasonable approximation. Indeed, since $\log \bar{p}(y_i \mid \bbeta)\leq \log p(y_i \mid \bbeta)$ for every $\xi_i$ and $i=1, \ldots, n$, the same holds for $\log \bar{p}(\by)$ and $\log p(\by)$. Thus, since $\log p(\by)$ does not vary with $\bxi$, maximizing $\log \bar{p}(\by)$ with respect to $\bxi$ is expected to provide the tightest approximation of each $\log p(y_i \mid \bbeta)$ via the lower bound in \eqref{eq5} evaluated at the optimum $\xi^*_i$, for $i=1, \ldots, n$, thereby guaranteeing similar predictive densities $p(\by)$ and $\bar{p}^*(\by)$. Hence, in correspondence to $\bxi^{*}$, the minimization of $\textsc{kl}[q(\bbeta) \mid \mid \bar{p}^{*}(\bbeta \mid \by) ]$ in the \textsc{e}--step, would hopefully provide  a solution ${q}^{*}(\bbeta)=\bar{p}^*(\bbeta \mid \by)$ close to the true posterior $p(\bbeta \mid y)$.

Although the above discussion provides an intuition for the excellent performance of the methods proposed by  \citet{jak_2000}, it shall be noticed that finding the tightest bound within a class of functions might not be sufficient if this class is not  flexible enough. Indeed, the quadratic form of \eqref{eq5} might be restrictive for logistic log-likelihoods, and hence even the optimal approximation may fail to mimic $\log p(y_i \mid \bbeta)$. Moreover, according to \eqref{eq1}, a formal \textsc{vb} set--up requires the minimization of a well--defined $\textsc{kl}$ divergence between an exact posterior and an approximating density from a given variational family. Instead, \citet{jak_2000} seem to minimize the divergence between an approximate posterior and a pre--specified density. If this were the case, then their methods could be only regarded as approximate solutions to formal \textsc{vb}. Indeed, although \eqref{eq5} has been recently studied  \citep[][]{de_2009,brow_2015}, this is currently the main view of the \textsc{em} in Algorithm \ref{al1}  \citep[e.g.][]{blei_2017, wang_2013, bishop_2006}. 

In Section~\ref{sec2} we prove that this is not true and that \eqref{eq5}, although apparently supported  by purely mathematical arguments, has indeed a clear probabilistic interpretation related to a recent P\'olya-gamma data augmentation  for logistic regression \citep{pol_2013}. In particular, let $q(z_i)$ be the density of a P\'olya-gamma $\textsc{pg}(1, \xi_i)$, then \eqref{eq5} is a proper evidence lower bound associated with a \textsc{vb} approximation of the posterior for $z_i$ in the conditional model $p(y_i, z_i \mid \bbeta)$ for data $y_i$ from \eqref{eq4} and the P\'olya-gamma variable $(z_i \mid \bbeta) \sim \textsc{pg}(1, \bx_i^{\intercal} \bbeta)$, with $\bbeta$ kept fixed. Combining this result with the objective function in equation \eqref{eq6}, allows us to formalize Algorithm \ref{al1} as a pure \textsc{cavi} which approximates the joint posterior of $\bbeta$ and the augmented P\'olya-gamma data $z_1, \ldots, z_n$, under a mean--field variational approximation within a conditionally conjugate exponential family framework. These results are discussed in Section \ref{sec3}, and are further generalized to allow stochastic variational inference  \citep{hoffman_2013} in logistic models, thus covering an important computational gap. A final discussion can be found in Section \ref{sec4}. Codes and additional empirical assessments are available at \url{https://github.com/tommasorigon/logisticVB}. Although we focus on Bayesian inference, it shall be noticed that \eqref{eq5} motivates also an \textsc{em} for maximum likelihood estimation of $\bbeta$ \citep[][Appendix C]{jak_2000}.  We derive the optimality properties of this routine in  the Appendix.

\section{Conditionally conjugate variational representation}
\label{sec2}
This section discusses the theoretical connection between equation \eqref{eq5} and a recent P\'olya-gamma data augmentation for conditionally conjugate inference in Bayesian logistic regression \citep{pol_2013}, thus allowing us to recast the methods proposed by  \citet{jak_2000} within the wider framework of mean--field variational inference for conditionally conjugate exponential family models.  We shall emphasize that, in a recent manuscript, \citet{scott_2013} proposed an \textsc{em} for maximum a posteriori estimation of $\bbeta$ in \eqref{eq4}, discussing connections with the variational methods in \citet{jak_2000}. Their findings are however limited to computational differences and similarities among the two methods and the associated algorithms. We instead provide a fully probabilistic connection between the contribution by  \citet{jak_2000} and the one of  \citet{pol_2013}, thus opening new avenues for advances in \textsc{vb} for logistic models.

To anticipate Lemma \ref{lemma1}, note that the core contribution of \citet{pol_2013} is in showing that $p(y_i \mid \bbeta)$ in model \eqref{eq4} can be expressed as a scale--mixture of Gaussians with respect to a P\'olya-gamma density. This result facilitates the implementation of \textsc{mcmc} methods which update  $\bbeta$ and the P\'olya-gamma augmented data $\bz=(z_1, \ldots, z_n)^{\intercal}$ from conjugate full conditionals. In fact, the joint density $p(\by,\bz\mid  \bbeta)$ has a Gaussian kernel in $\bbeta$, thus restoring Gaussian--Gaussian conjugacy in the full conditional. As discussed in  Lemma \ref{lemma1}, this data augmentation, although developed a decade later, was implicitly hidden in the bound of  \citet{jak_2000}.
\begin{lemma}
\label{lemma1}
Let $\log \bar{p}(y_i \mid \bbeta)$ be the quadratic lower bound in \eqref{eq5} proposed by \citet{jak_2000} for the logistic log-likelihood $\log p(y_i \mid \bbeta)$  in \eqref{eq4}. Then, for every unit $i=1, \ldots, n$, we have 
\begin{eqnarray}
\log \bar{p}(y_i \mid \bbeta)=\int_{\Re_+} q(z_i) \log \frac{p(y_i, z_i \mid \bbeta)}{q(z_i)} \mbox{d}z_i=\mathbb{E}_{q(z_i)}[\log p(y_i, z_i \mid \bbeta)]-\mathbb{E}_{q(z_i)}[\log q(z_i)], 
\label{eq7}
\end{eqnarray}
with $p(y_i, z_i \mid \bbeta)=p(y_i \mid \bbeta)p(z_i \mid \bbeta)$ and $p(y_i \mid \bbeta)=\exp(y_i\bx_i^{\intercal} \bbeta)[1+\exp(\bx_i^{\intercal} \bbeta)]^{-1}$, whereas $q(z_i)$ and $p(z_i \mid \bbeta)$ are the densities of the P\'olya-gamma  variables 
 $\textsc{pg}(1,\xi_i)$ and $\textsc{pg}(1,\bx_i^{\T}\bbeta)$, respectively.
\end{lemma}
\begin{proof}[Proof]
To prove Lemma~\ref{lemma1}, first notice that $0.5\xi_i+\log[1+\exp(-\xi_i)]=\log[2\mbox{cosh}(0.{5}\xi_i)]$ and $0.{5}(\bx_i^{\T}\bbeta)=\log[1+\exp(\bx_i^{\T}\bbeta)]-\log[2\mbox{cosh}\{0.5(\bx_i^{\T}\bbeta)\}]$. Replacing such quantities in   \eqref{eq5}, we obtain
\begin{eqnarray*}
y_i \bx_i^{\intercal}\bbeta-\log[1+\exp(\bx_i^{\T}\bbeta)]-0.25\xi_i^{-1}\mbox{\textup{tanh}}(0.5\xi_i)[(\bx_i^{\T}\bbeta)^2- \xi_i^2]+\log[\mbox{cosh}(0.{5}\xi_i)^{-1}\mbox{cosh}\{0.5(\bx_i^{\T}\bbeta)\}].
\end{eqnarray*}
To highlight equation \eqref{eq7} in the above function, note that, recalling \citet{pol_2013}, the quantity $- 0.25\xi_i^{-1}\mbox{\textup{tanh}}(0.5\xi_i)[(\bx_i^{\T}\bbeta)^2- \xi_i^2]$ is equal to $\mathbb{E}[-0.5z_{i}(\bx_i^{\T}\bbeta)^2]-\mathbb{E}(-0.5z_{i}\xi_i^2)$, where the expectation is taken with respect to $z_i \sim \textsc{pg}(1,\xi_i)$. Hence, $\log \bar{p}(y_i \mid \bbeta)$ can be expressed as
\begin{eqnarray*}
\begin{split}
\int_{\Re_+}\frac{\exp(-0.5z_i\xi_i^2)p(z_i)}{\mbox{cosh}(0.{5}\xi_i)^{-1}}  \log \frac{\exp(y_i \bx_i^{\intercal}\bbeta)[1+\exp(\bx_i^{\T}\bbeta)]^{-1}\exp[-0.5z_i(\bx_i^{\T}\bbeta)^2]\mbox{cosh}[0.{5}(\bx_i^{\T}\bbeta)]p(z_i)}{\exp(-0.5z_i\xi_i^2)\mbox{cosh}(0.{5}\xi_i)p(z_i)} \mbox{d}z_i.
 \end{split}
\end{eqnarray*}
Based on the above expression, the proof is concluded after noticing that $\exp(y_i \bx_i^{\intercal}\bbeta)[1+\exp(\bx_i^{\T}\bbeta)]^{-1}=p(y_i \mid \bbeta)$, whereas $\exp[-0.5z_i(\bx_i^{\T}\bbeta)^2]\mbox{cosh}[0.{5}(\bx_i^{\T}\bbeta)]p(z_i)$ and $\exp(-0.5z_i\xi_i^2)\mbox{cosh}(0.{5}\xi_i)p(z_i)$ are the densities  $p(z_i \mid \bbeta)$ and $q(z_i)$ of the  P\'olya-gamma random variables $\textsc{pg}(1,\bx_i^{\T}\bbeta)$ and $\textsc{pg}(1,\xi_i)$, respectively, with $p(z_i)$ the density of a $\textsc{pg}(1,0)$.
\end{proof}

According to Lemma \ref{lemma1}, the expansion in equation \eqref{eq5} is a proper \textsc{elbo} related to a \textsc{vb} approximation of the posterior for $z_i$ in the conditional model $p(y_i, z_i \mid \bbeta)$ for response data $y_i$ from \eqref{eq4} and the local variable $(z_i \mid \bbeta) \sim \textsc{pg}(1, \bx_i^{\intercal} \bbeta)$, with $\bbeta$ kept fixed. Note that, although some intuition on the relation between $\log \bar{p}(y_i \mid \bbeta)$ and $\mathbb{E}_{q(z_i)}[\log p(y_i, z_i \mid \bbeta)]$ can be deduced from  \citet{scott_2013}, the authors leave out additive constants not depending on $\bbeta$ in $\log \bar{p}(y_i \mid \bbeta)$ when discussing this connection. Indeed, according to Lemma \ref{lemma1}, these quantities are crucial to formally interpret $\log \bar{p}(y_i \mid \bbeta)$ as a genuine \textsc{elbo}, since they coincide with $-\mathbb{E}_{q(z_i)}[\log q(z_i)]$. Besides this result,  Lemma~\ref{lemma1} provides a formal characterization for the approximation error  $\log p(y_i \mid \bbeta) - \log \bar{p}(y_i \mid \bbeta)$. Indeed, adapting  \eqref{eq2} to this setting, such a quantity is   the  \textsc{kl} divergence between a generic P\'olya-gamma variable  and the one obtained by conditioning on $\bbeta$. This allows to complete $\log p(y_i \mid \bbeta) \geq \log \bar{p}(y_i \mid \bbeta)$, as
\begin{eqnarray}
\quad \log p(y_i \mid \bbeta)= \log \bar{p}(y_i \mid \bbeta)+\textsc{kl}[q(z_i) \mid \mid p(z_i \mid y_i, \bbeta)]=\log \bar{p}(y_i \mid \bbeta)+\textsc{kl}[q(z_i) \mid \mid p(z_i \mid \bbeta)],
\label{eq8}
\end{eqnarray}
where the last equality follows from the fact that $p(y_i,z_i {\mid} \bbeta)=p(y_i {\mid} \bbeta)p(z_i {\mid} \bbeta)$, and hence $p(z_i {\mid} y_i, \bbeta)=p(z_i {\mid} \bbeta)$. This result  sheds light on the heuristic interpretation of ${q}^{*}(\bbeta)$ in Section \ref{sec1}. Indeed, as is clear from  \eqref{eq8}, if $q(z_i)$ evaluated at the optimal $\xi^{*}_i$ is globally close to $p(z_i \mid \bbeta)$ for every $\bbeta$ and $i=1, \ldots, n$, then  \eqref{eq5} ensures accurate approximation of $\log p(y_i \mid \bbeta)$, thus providing approximate posteriors  ${q}^{*}(\bbeta)$  close to the target $p(\bbeta \mid \by)$. Exploiting Lemma \ref{lemma1}, Theorem \ref{teo1} formalizes this discussion by proving that the \textsc{em} in Algorithm \ref{al1} maximizes the \textsc{elbo} of a well--defined model under a  mean--field \textsc{vb}. 
\begin{theorem}
\label{teo1}
The lower bound in \eqref{eq6} maximized by \citet{jak_2000} in their \textsc{em} for approximate Bayesian inference in model \eqref{eq4} coincides with a genuine evidence lower bound 
\begin{eqnarray}
\begin{split}
\textsc{elbo}[q(\bbeta,\bz)]=&\int_{\Re^p}\int_{\Re_+^n}q(\bbeta, \bz) \log \frac{p(\by, \bbeta,\bz)}{q(\bbeta, \bz)} \mbox{\normalfont d}\bz \mbox{\normalfont  d} \bbeta,\\
=&\ \mathbb{E}_{q(\bbeta,\bz)}[\log p(\by, \bbeta,\bz)]-\mathbb{E}_{q(\bbeta,\bz)}[\log q(\bbeta,\bz)],
\end{split}
\label{eq9}
\end{eqnarray}
where $p(\by, \bbeta,\bz)=p(\bbeta)\prod_{i=1}^np(y_i \mid \bbeta)p(z_i \mid \bbeta)$ and $q(\bbeta, \bz)=q(\bbeta)\prod_{i=1}^n q(z_i)$, with $q(z_i)$ and $p(z_i \mid \bbeta)$ denoting the densities of  the P\'olya-gamma variables $\textsc{pg}(1,\xi_i)$ and  $\textsc{pg}(1,\bx_i^{\T}\bbeta)$, respectively. 
\end{theorem}
\begin{proof}[Proof]
The proof is a direct consequence of  Lemma~\ref{lemma1}. In particular, let $\int_{\Re^p} q(\bbeta)\log[ p(\bbeta)q(\bbeta)^{-1}] \mbox{d} \bbeta+ \int_{\Re^p} q(\bbeta)\sum_{i=1}^n \log \bar{p}(y_i \mid \bbeta)  \mbox{d} \bbeta$ denote an expanded representation of \eqref{eq6}. Then, replacing $\log \bar{p}(y_i \mid \bbeta)$ with its probabilistic definition in \eqref{eq7} and performing simple mathematical calculations, we obtain
\begin{eqnarray*}
\int_{\Re^p} q(\bbeta)\log \frac{p(\bbeta)}{q(\bbeta)} \mbox{d} \bbeta+\sum_{i=1}^n\int_{\Re^p}\int_{\Re_+}  q(\bbeta) q(z_i) \log \frac{p(y_i \mid \bbeta)p(z_i \mid \bbeta)}{q(z_i)} \mbox{d}z_i  \mbox{d} \bbeta.
\end{eqnarray*}
Note now that the first summand does not depend on $\bz$, thus allowing us to replace this integral  with $\int_{\Re^p}\int_{\Re_+^n}\log [p(\bbeta)q(\bbeta)^{-1}]q(\bbeta) \prod_{i=1}^n q(z_i)\mbox{d}\bz\mbox{d}\bbeta$. Similar arguments can be made to include $\prod_{i=1}^n q(z_i)$ in the second integral. Making these substitutions in the above equation we obtain
\begin{eqnarray*}
\begin{split}
&{\int_{\Re^p} \int_{\Re^n_+}} \left[\log \frac{p(\bbeta)}{q(\bbeta)}+\log \frac{\prod_{i=1}^np(y_i\mid  \bbeta)p(z_i\mid \bbeta)}{\prod_{i=1}^n q(z_i) } \right]q(\bbeta)\prod_{i=1}^n q(z_i)\mbox{\normalfont d}\bz \mbox{\normalfont  d} \bbeta \\
& {=}{\int_{\Re^p}\int_{\Re_+^n}}q(\bbeta, \bz) \log \frac{p(\bbeta) \prod_{i=1}^np(y_i \mid \bbeta)p(z_i\mid \bbeta)}{q(\bbeta, \bz)} \mbox{\normalfont d}\bz \mbox{\normalfont  d} \bbeta=\int_{\Re^p}\int_{\Re_+^n}q(\bbeta, \bz) \log \frac{p(\by, \bbeta,\bz)}{q(\bbeta, \bz)} \mbox{\normalfont d}\bz \mbox{\normalfont  d} \bbeta,
\end{split}
\end{eqnarray*}
thus proving Theorem~\ref{teo1}. Note that $q(\bbeta, \bz )=q(\bbeta)\prod_{i=1}^n q(z_i)$ and $\int_{\Re_+} q(z_i)\mbox{d}z_i=1$.
\end{proof}
As is clear from Theorem~\ref{teo1}, the variational strategy proposed by \citet{jak_2000} is a pure \textsc{vb}  minimizing $\textsc{kl}[q(\bbeta, \bz) \mid \mid p(\bbeta,\bz \mid \by)]$ under a mean--field variational family $\mathcal{Q}=\{ q(\bbeta, \bz): q(\bbeta, \bz)=q(\bbeta) \prod_{i=1}^nq(z_i) \}$ in the conditionally conjugate exponential family model with 
\begin{eqnarray}
\begin{split}
&\mbox{Global variables} \qquad \quad \qquad \bbeta \sim \mbox{N}_p(\bmu_0,\bSigma_0),\\
&\mbox{Local variables}  \qquad \ \ \ (z_i \mid \bbeta) \sim \textsc{pg}(1,\bx_i^{\intercal}\bbeta ), \quad i=1, \ldots, n, \\
& \mbox{Data} \qquad \qquad \qquad \quad \ (y_i \mid \bbeta) \sim \mbox{Bern}\{\exp(\bx_i^{\intercal}\bbeta)[1+\exp(\bx_i^{\intercal}\bbeta)]^{-1}\}, \quad i=1, \ldots, n.
\end{split}
\label{eq10}
\end{eqnarray}
We refer to \citet[][Sect. 2]{choi_2013} for this specific formulation of the P\'olya-gamma data augmentation scheme which highlights how, unlike the general specification in \eqref{eq3}, the conditional distribution of $y_i$ does not depend on $z_i$. As discussed in Section \ref{sec1}, this is not a necessary requirement. Indeed, what is important is that the joint likelihood $p(y_i, z_i \mid \bbeta)$ is within an exponential family and the prior $p(\bbeta)$ is conjugate to it.  Recalling  \citet[][Sect. 2]{choi_2013} and noticing that $\mbox{cosh}[0.{5}(\bx_i^{\T}\bbeta)]=0.5[1+\exp(\bx_i^{\T}\bbeta)]\exp[-0.{5}(\bx_i^{\T}\bbeta)]$, this is the case of \eqref{eq10}. In fact
\begin{eqnarray}
p(y_i, z_i \mid \bbeta)&=&p(y_i\mid \bbeta)p(z_i\mid \bbeta)=\exp(\bx_i^{\intercal} \bbeta)^{y_i}[1+\exp(\bx_i^{\intercal} \bbeta)]^{-1}\exp[-0.5z_i(\bx_i^{\T}\bbeta)^2]\mbox{cosh}[0.{5}(\bx_i^{\T}\bbeta)]p(z_i),\nonumber\\
&=&0.5\exp[(y_i-0.5)\bx_i^{\T}\bbeta-0.5z_i(\bx_i^{\T}\bbeta)^2][1+\exp(\bx_i^{\intercal} \bbeta)]^{-1}[1+\exp(\bx_i^{\intercal} \bbeta)]p(z_i),
\label{gaus_li}
\end{eqnarray}
is proportional to the Gaussian kernel $\exp[(y_i-0.5)\bx_i^{\T}\bbeta-0.5z_i(\bx_i^{\T}\bbeta)^2]$, which is conjugate to $p(\bbeta)$.

\section{\textsc{cavi} and \textsc{svi} for logistic models}
\label{sec3}
The results in Section \ref{sec2} recast the methods in \citet{jak_2000}  within a much broader framework motivating a formal \textsc{cavi} and generalizations to stochastic variational inference (\textsc{svi}). 

\subsection{Coordinate ascent variational inference (\textsc{CAVI})}
\label{seccavi}
As discussed in Section \ref{sec1}, the mean--field assumption allows the implementation of a simple  \textsc{cavi}  \citep[Ch. 10.1.1]{blei_2017,bishop_2006} which sequentially maximizes the evidence lower bound in \eqref{eq9}  with respect to each factor in $q(\bbeta) \prod_{i=1}^nq(z_i)$, via the following updates
\begin{eqnarray}
\begin{split}
q^{(t)}(\bbeta)  =& \ \exp\{\mathbb{E}_{q^{(t-1)}(\bz)}\log[ p(\bbeta\mid \by,\bz)]\}c_{\bbeta}(\by)^{-1}, \\
 q^{(t)}(z_i) =&\ \exp\{\mathbb{E}_{q^{(t)}(\bbeta)}\log[ p(z_i\mid \by, \bz_{-i},\bbeta)]\}c_{z_i}(\by)^{-1}, \quad i=1, \ldots, n,
\end{split}
\label{eq11}
\end{eqnarray}
 at iteration $t$---until convergence of the \textsc{elbo}. In the above expressions, $c_{\bbeta}(\by)$ and $c_{z_i}(\by)$, $i=1, \ldots, n$, denote constants leading to proper densities. Note that in our case $p(z_i\mid \by, \bz_{-i},\bbeta)=p(z_i\mid \by, \bbeta)$. 
 
To clarify why \eqref{eq11} provides a routine which iteratively improves the \textsc{elbo}, and ultimately maximizes it, note that, keeping fixed $q^{(t-1)}(z_1), \ldots, q^{(t-1)}(z_n)$, equation \eqref{eq9} can be re--written as
\begin{eqnarray}
\begin{split}
& \mathbb{E}_{q(\bbeta)}\left[\mathbb{E}_{q^{(t-1)}(\bz)}\log\frac{ p(\bbeta) \prod\nolimits_{i=1}^np(y_i\mid \bbeta)p(z_i\mid \bbeta)p(\by,\bz)}{q(\bbeta)p(\by,\bz)}\right]+\mbox{const} \\
&=  \mathbb{E}_{q(\bbeta)}\left[\mathbb{E}_{q^{(t-1)}(\bz)}\log\frac{ p(\bbeta \mid \by,\bz)}{q(\bbeta)}\right]+ \mathbb{E}_{q^{(t-1)}(\bz)}\log p(\by,\bz)+\mbox{const},\\
&= \mathbb{E}_{q(\bbeta)}\left[\log\frac{\exp\{\mathbb{E}_{q^{(t-1)}(\bz)}\log[ p(\bbeta\mid \by,\bz)]\}}{q(\bbeta)c_{\bbeta}(\by)}\right]+\log c_{\bbeta}(\by) + \mathbb{E}_{q^{(t-1)}(\bz)}\log p(\by,\bz)+\mbox{const},
\end{split}
\label{eq11b}
\end{eqnarray} 
where the first term  in the last equation is the only quantity which depends on $\bbeta$ and is equal to the negative $\textsc{kl}$ divergence between   $q(\bbeta)$ and $\exp\{\mathbb{E}_{q^{(t-1)}(\bz)}\log[ p(\bbeta\mid \by,\bz)]\}c_{\bbeta}(\by)^{-1}$, thus motivating the \textsc{cavi} update for $q(\bbeta)$. Similar derivations can be done to obtain the solutions for $q(z_1), \ldots, q(z_n)$ in \eqref{eq11}. As is clear from \eqref{eq11}, the \textsc{cavi} solution  identifies both the form of the approximating densities---without pre--specifying them as part of the mean--field assumption---and the optimal parameters of such densities. As discussed in Section \ref{sec1}, these solutions are particularly straightforward in conditionally conjugate exponential family representations \citep{hoffman_2013}, including model \eqref{eq10}. In fact,  recalling \citet{pol_2013}, the full conditionals for the local and global variables in model  \eqref{eq10} can be obtained via conditional conjugacy properties, which lead to
\begin{eqnarray}
\begin{split}
 (\bbeta\mid \by,\bz) &\sim \mbox{N}_{p}\{(\bSigma_0^{-1}+\bX^{\intercal} \bZ \bX)^{-1}[\bX^{\intercal}(\by-0.5{\cdot}{\bf 1}_n)+\bSigma_0^{-1}\bmu_0 ],(\bSigma_0^{-1}+\bX^{\intercal} \bZ \bX)^{-1}\},\\
\qquad (z_i \mid \by,\bz_{-i},\bbeta) &\sim \textsc{pg}(1, \bx_i^{\intercal}\bbeta), \qquad i=1, \ldots, n,
\end{split}
\label{eq12}
\end{eqnarray}
with $\bZ=\mbox{diag}(z_1, \ldots, z_n)$ and $\bX$ the $n \times p$ design matrix with rows $\bx_i^{\intercal}$, $i=1, \ldots, n$. Moreover, as is clear from \eqref{eq12}, both $(\bbeta\mid \by,\bz) $ and $(z_i \mid \by,\bz_{-i},\bbeta)$ have the exponential family representation 
\begin{eqnarray}
\begin{split}
p(\bbeta\mid \by,\bz) &=h(\bbeta)\exp[\boeta_1(\by)^{\intercal}\bbeta+\boeta_2(\bz)^{\intercal}\bbeta\bbeta^{\intercal}-\alpha\{\boeta_1(\by),\boeta_2(\bz) \} ],\\
p(z_i \mid \by,\bz_{-i},\bbeta) &=p(z_i \mid \by,\bbeta) =h(z_i)\exp[\eta_i(\bbeta)^{\intercal}z_i-\alpha\{\eta_i(\bbeta) \} ], \qquad i=1, \ldots, n,
\end{split}
\label{eq13}
\end{eqnarray}
with natural parameters $\boeta_1(\by)=\bX^{\intercal}(\by-0.5{\cdot}{\bf 1}_n){+}\bSigma_0^{-1}\bmu_0$, $\boeta_2(\bz) =-0.5(\bSigma_0^{-1}+\bX^{\intercal} \bZ \bX)$, and $\eta_i(\bbeta)=-0.5(\bx^{\intercal}_i\bbeta)^2$. Substituting these expressions in \eqref{eq11}, it can be immediately noticed that the \textsc{cavi} solutions  have the same density of the corresponding full--conditionals with optimal natural parameters $\blambda^{(t)}_1= \mathbb{E}_{q^{(t-1)}(\bz)}[\boeta_1(\by)]$, $\blambda^{(t)}_2=\mathbb{E}_{q^{(t-1)}(\bz)}[\boeta_2(\bz)]$ and $\phi^{(t)}_{i}=\mathbb{E}_{q^{(t)}(\bbeta)}[\eta_i(\bbeta)]$, $i=1, \ldots, n$. 

\begin{algorithm}[t]
\caption{\textsc{cavi} for logistic regression.} \label{al2}
Initialize $\xi_1^{(0)}, \ldots, \xi^{(0)}_n$.\\
 \For( ){{\normalfont{$t=1$ until convergence of the evidence lower bound $\textsc{elbo}[q(\bbeta,\bz)]$}}}
 {
\enspace   {\bf Maximization}. Maximize $\textsc{elbo}[q(\bbeta)\prod_{i=1}^nq^{(t-1)}(z_i)]$ with respect to $q(\bbeta)$. As discussed in Section \ref{seccavi}, this maximization provides a Gaussian density for $q^{(t)}(\bbeta)$ having natural parameters
\begin{eqnarray*}
\blambda^{(t)}_1=\mathbb{E}_{q^{(t-1)}(\bz)}[\boeta_1(\by)]=\bX^{\intercal}(\by-0.5{\cdot}{\bf 1}_n)+\bSigma_0^{-1}\bmu_0, \quad \blambda^{(t)}_2=\mathbb{E}_{q^{(t-1)}(\bz)}[\boeta_2(\bz)]=-0.5(\bSigma_0^{-1}+\bX^{\intercal} \bar{\bZ}^{(t-1)} \bX),
\end{eqnarray*}
with $\bar{\bZ}^{(t-1)}=\mbox{diag}\{0.5[\xi_1^{(t-1)}]^{-1}\mbox{tanh}(0.5 \xi^{(t-1)}_1), \ldots,  0.5[\xi_n^{(t-1)}]^{-1}\mbox{tanh}(0.5 \xi^{(t-1)}_n)\}$. Hence, the approximating density is that of a $\mbox{N}_p(\bmu^{(t)}, \bSigma^{(t)})$ with $ \bmu^{(t)}=[-2\blambda^{(t)}_2]^{-1} \blambda^{(t)}_1$ and $\bSigma^{(t)}=[-2\blambda^{(t)}_2]^{-1}$.\\
\vspace{3pt}
\enspace
{\bf Maximization}. Maximize $\textsc{elbo}[q^{(t)}(\bbeta)\prod_{i=1}^nq(z_i)]$ with respect to $\prod_{i=1}^n q(z_i)$. As discussed in Section \ref{seccavi}, this maximization provides a P\'olya-gamma density for each $q^{(t)}(z_i)$, $i=1, \ldots, n$, having natural parameters
\begin{eqnarray*}
\phi^{(t)}_{i}=\mathbb{E}_{q^{(t)}(\bbeta)}[\eta_i(\bbeta)]=-0.5[\bx^{\intercal}_i\bSigma^{(t)}\bx_i{+}(\bx_i^{\intercal}\bmu^{(t)})^2], \quad i=1, \ldots, n.
\end{eqnarray*}
Thus, each $q^{(t)}(z_i)$ is the density of a $\textsc{pg}(1, \xi^{(t)}_i)$ with parameter $\xi^{(t)}_i=[-2\phi^{(t)}_{i}]^{\frac{1}{2}}$ for $i=1,\dots,n$. Note that $\xi_i$ and $-\xi_i$ induce the same P\'olya-gamma density. Hence, there is no ambiguity in the above square root. A similar remark, from a different perspective, is found in footnote 3 of \citet{jak_2000}.
}
{\bf Output at the end of the algorithm:} $q^*(\bbeta,\bz)=q^*(\bbeta)\prod_{i=1}^n q^*(z_i)$. 
\end{algorithm}

As shown in Algorithm \ref{al2}, the above expectations can be computed in closed--form since $q(\bbeta)$ and $q(z_1), \ldots, q(z_n)$ are already known to be Gaussian and P\'olya-gammas, thus requiring only sequential optimizations of  natural parameters. This form of \textsc{cavi}, which is discussed  in \citet{hoffman_2013} and is known in the  literature as variational Bayesian \textsc{em} \citep{beal_2003}, clarifies the link between \textsc{cavi} and the \textsc{em} in \citet{jak_2000}. Indeed, recalling Section \ref{sec2}, both methods optimize the same objective function and rely, implicitly, on the same steps. In particular, due to Lemma \ref{lemma1}, the \textsc{e}--step in Algorithm \ref{al1} is in fact maximizing the conditional $\textsc{elbo}[q(\bbeta)\prod_{i=1}^nq^{(t-1)}(z_i)]$ with respect to $q(\bbeta)$ as in the first maximization of Algorithm \ref{al2}. Similarly, the \textsc{m}--step solution for $\bxi$ in Algorithm \ref{al1} is actually the one maximizing the conditional $\textsc{elbo}[q^{(t)}(\bbeta)\prod_{i=1}^nq(z_i)]$ with respect to $\prod_{i=1}^n q(z_i)$ in the second optimization of the \textsc{cavi} in Algorithm \ref{al2}.  

\subsection{Stochastic variational inference (SVI)}
\label{svi}
Algorithm \ref{al2} and model \eqref{eq10} motivate further generalizations in large $n$ studies when \textsc{cavi} can face severe computational bottleneck. Indeed, each iteration of Algorithm \ref{al2} requires optimization of the whole local natural parameters $\phi_{i}$, $i=1, \ldots, n$ and summation over the entire dataset when updating $(\blambda_1, \blambda_2)$. This issue has been addressed by \citet{hoffman_2013} via computationally cheaper updates under a \textsc{svi} routine for scalable mean--field \textsc{vb}  in conditionally conjugate exponential family models. Leveraging the probabilistic results in Section \ref{sec2}, we adapt this strategy to Bayesian logistic regression, thus covering an important computational gap.

To clarify the core idea underlying \textsc{svi}, note that, joining equations \eqref{eq11}--\eqref{eq13} and recalling \citet[][Sect 2.2]{hoffman_2013}, the \textsc{cavi} solutions for the parameters $(\blambda_1, \blambda_2)$ at iteration $t$ are indeed solving the maximization problem $\mbox{argmax}_{\blambda_1,\blambda_2}(\mathbb{E}_{q(\bbeta)}\{\mathbb{E}_{q^{(t-1)}(\bz)}\log[ p(\bbeta\mid \by,\bz)]\}-\mathbb{E}_{q(\bbeta)}\log q(\bbeta)+\mbox{const}),$
where $p(\bbeta\mid \by,\bz)$ and $q(\bbeta)$ have the same exponential family representation with natural parameters $[\boeta_1(\by),\boeta_2(\bz)]$ and $\blambda=(\blambda_1,\blambda_2)$, respectively. Adapting \citet[][Sect 2.2]{hoffman_2013} to our case, this optimization can be solved by equating to 0 the gradient 
\vspace{-5pt}
\begin{eqnarray*}
\nabla_{\blambda}[\mathbb{E}_{q(\bbeta)}\log p(\bbeta)-\mathbb{E}_{q(\bbeta)}\log q(\bbeta)] +{\sum\nolimits_{i=1}^{n}}{\nabla_{\blambda}}\mathbb{E}_{q(\bbeta)}[\mathbb{E}_{q^{(t-1)}(z_i)} \log p(y_i,z_i \mid \bbeta)]\qquad \qquad \qquad \qquad \qquad \qquad  &&\\
{=}\nabla_{\blambda}[\mathbb{E}_{q(\bbeta)}\log p(\bbeta){-}\mathbb{E}_{q(\bbeta)}\log q(\bbeta)] {+}{\sum\nolimits_{i=1}^{n}} \{(y_i{-}0.5){\nabla_{\blambda}}\mathbb{E}_{q(\bbeta)}(\bx_i^{\intercal}\bbeta){-}0.5 \mathbb{E}_{q^{(t-1)}(z_i)}(z_i){\nabla_{\blambda}}\mathbb{E}_{q(\bbeta)}[(\bx_i^{\intercal}\bbeta)^{2}]\} \qquad  &&
\end{eqnarray*}
thus obtaining the estimating equations
\begin{eqnarray}
\mathbb{E}_{q^{(t-1)}(\bz)}[\boeta_1(\by)]-\blambda_1=0, \qquad \mathbb{E}_{q^{(t-1)}(\bz)}[\boeta_2(\bz)]-\blambda_2=0, 
\label{eq14}
\end{eqnarray}
whose solution provides the optimal parameters $\blambda^{(t)}_1$ and $\blambda^{(t)}_2$ from the \textsc{cavi}.

\begin{algorithm}[t]
\caption{\textsc{svi} for logistic regression.} \label{al3}
Initialize $(\blambda_1^{(0)},\blambda_2^{(0)})$ randomly and set the step-size sequence $\rho_t$ appropriately.\\
 \For( ){{\normalfont{$t=1$ until a large number of iterations}}}
 {
  \vspace{3pt}
 \enspace   {\bf Sampling.} Sample a data point $(y_i, \bx_i)$ randomly from the dataset.\\
 \vspace{3pt}
 \enspace
{\bf Local Maximization.} Calculate the locally optimized density for $z_i$ as a function of the latest value $\blambda^{(t-1)}$ for $\blambda$, thus obtaining a P\'olya-gamma with natural parameter
\begin{eqnarray*}
\phi_{i}(\blambda^{(t-1)})=-0.5[\bx^{\intercal}_i\bSigma^{(t-1)}\bx_i{+}(\bx_i^{\intercal}\bmu^{(t-1)})^2]=-0.5\{\bx^{\intercal}_i[-2\blambda^{(t-1)}_2]^{-1}\bx_i{+}(\bx_i^{\intercal}[-2\blambda^{(t-1)}_2]^{-1} \blambda^{(t-1)}_1)^2\}.
\end{eqnarray*}
Therefore, the locally optimized density is the one of a $\textsc{pg}[1, \xi_i(\blambda^{(t-1)})]$ with $\xi_i(\blambda^{(t-1)})=[-2\phi_{i}(\blambda^{(t-1)})]^{\frac{1}{2}}$.\\
\vspace{3pt}
\enspace   {\bf Global parameters updates}. Update the global parameters according to the \citet{robbins_1951} iterative procedure outlined in \eqref{eq17}. This provides the solutions
\begin{eqnarray*}
\blambda^{(t)}_1=(1-\rho_t)\blambda^{(t-1)}_1+\rho_t[n\bx_i(y_i-0.5)+\bSigma_0^{-1}\bmu_0], \quad \blambda^{(t)}_2=(1-\rho_t)\blambda^{(t-1)}_2-\rho_t0.5(\bSigma_0^{-1}+n\bx_i \bar{z_i}^{(t-1)} \bx^{\intercal}_i),
\end{eqnarray*}
with $\bar{z}^{(t-1)}=0.5[\xi_i(\blambda^{(t-1)})]^{-1}\mbox{tanh}[0.5 \xi_i(\blambda^{(t-1)})]$. Hence, the approximating density is that of a Gaussian random variable with expectation $ \bmu^{(t)}=[-2\blambda^{(t)}_2]^{-1} \blambda^{(t)}_1$ and variance--covariance matrix $\bSigma^{(t)}=[-2\blambda^{(t)}_2]^{-1}$.\\
}
{\bf Output at the end of the algorithm:} $q_{\textsc{svi}}^*(\bbeta)$. 
\end{algorithm}

Motivated by this alternative view of  \textsc{cavi}, \citet{hoffman_2013} proposed a scalable \textsc{svi} routine relying on stochastic optimization \citep{robbins_1951} of the \textsc{elbo} in \eqref{eq9} as a direct function of the global parameters $\blambda$. Specifically, let $q(\bbeta)$  be the Gaussian approximating distribution parameterized by $\blambda$, and $q_{\mbox{\tiny opt}}(z_1), \ldots, q_{\mbox{\tiny opt}}(z_n)$ the P\`olya-gamma densities with optimal natural parameters $\phi_1(\blambda)=\mathbb{E}_{q(\bbeta)}[\eta_1(\bbeta)], \ldots, \phi_n(\blambda)=\mathbb{E}_{q(\bbeta)}[\eta_n(\bbeta)]$, then optimizing the locally maximized \textsc{elbo}
\begin{eqnarray}
 \mathcal{L}(\blambda)&=&\int_{\Re^p}\int_{\Re_+^n}q(\bbeta) \prod_{i=1}^{n}q_{\mbox{\tiny opt}}(z_i)\left[\log \frac{p(\bbeta) \prod_{i=1}^np(y_i\mid \bbeta)p(z_i \mid \bbeta)}{q(\bbeta) \prod_{i=1}^{n}q_{\mbox{\tiny opt}}(z_i)} \right] \mbox{\normalfont  d} \bz \mbox{\normalfont  d} \bbeta, \nonumber \\
&=&\mathbb{E}_{q(\bbeta)}\log p(\bbeta)-\mathbb{E}_{q(\bbeta)}\log q(\bbeta)+\sum\nolimits_{i=1}^{n}\mathbb{E}_{q(\bbeta)}(\mathbb{E}_{q_{\mbox{\tiny opt}}(z_i)}\{\log[p(y_i \mid \bbeta)p(z_i \mid \bbeta)]-\log q_{\mbox{\tiny opt}}(z_i)\}),
\label{eq15}
\end{eqnarray}
provides an optimal solution for the global parameters $\blambda$ and, as a direct consequence, for the locally optimized ones $\phi_1(\blambda), \ldots, \phi_n(\blambda)$. This ensures maximization of  \eqref{eq9}. Before deriving the  \textsc{svi} routine, let us first highlight a key connection between the \textsc{cavi} solutions in \eqref{eq14} and those arising from the optimization of $ \mathcal{L}(\blambda)$. To do this, note that recalling Lemma \ref{lemma1} and its proof, the functions within the summation term in \eqref{eq15} coincide with the expectations of the conditional \textsc{elbo}s in \eqref{eq7} evaluated at the optimal P\`olya-gamma densities with $\xi_i(\blambda)=[-2 \phi_{i}(\blambda)]^{1/2}$, thus providing
\begin{eqnarray*}
\begin{split}
&\mathbb{E}_{q(\bbeta)}(\mathbb{E}_{q_{\mbox{\tiny opt}}(z_i)}\{\log[p(y_i \mid \bbeta)p(z_i \mid \bbeta)]-\log q_{\mbox{\tiny opt}}(z_i)\})  \\
&=(y_i-0.5)\mathbb{E}_{q(\bbeta)}(\bx_i^{\intercal}\bbeta)-0.5\mathbb{E}_{q_{\mbox{\tiny opt}}(z_i)}(z_i)\{\mathbb{E}_{q(\bbeta)}[(\bx_i^{\intercal}\bbeta)^2]+2\phi_i(\blambda)\}-\log \mbox{cosh}\{0.5[-2\phi_i(\blambda)]^{1/2}\}+\mbox{const},  \qquad\\
&=(y_i-0.5)\mathbb{E}_{q(\bbeta)}(\bx_i^{\intercal}\bbeta)+\alpha[\phi_i(\blambda)]+\mbox{const}, \quad i=1, \ldots, n,
\end{split}
\end{eqnarray*} 
where the last equality follows after noticing that $\phi_i(\blambda)=\mathbb{E}_{q(\bbeta)}[\eta_i(\bbeta)]=\mathbb{E}_{q(\bbeta)}[-0.5(\bx^{\intercal}_i\bbeta)^2]$ and that $-\log \mbox{cosh}\{0.5[-2\phi_i(\blambda)]^{1/2}\}$ is the function $\alpha[\phi_i(\blambda)]$ in the exponential family representation for the density of the P\`olya-gamma with parameters $1$ and $[-2 \phi_{i}(\blambda)]^{1/2}$. Since our final goal is to maximize $\mathcal{L}(\blambda)$, let us substitute the above equation in \eqref{eq15} and compute $\nabla_{\blambda}\mathcal{L}(\blambda)$. This  leads to
\begin{eqnarray*}
\nabla_{\blambda}[\mathbb{E}_{q(\bbeta)}\log p(\bbeta)-\mathbb{E}_{q(\bbeta)}\log q(\bbeta)] +{\sum\nolimits_{i=1}^{n}}\{(y_i-0.5)\nabla_{\blambda}\mathbb{E}_{q(\bbeta)}(\bx_i^{\intercal}\bbeta)+\nabla_{\phi_i(\blambda)}\alpha[\phi_i(\blambda)]\nabla_{\blambda}\phi_i(\blambda)\}\quad \quad \quad &&\\
=\nabla_{\blambda}[\mathbb{E}_{q(\bbeta)}\log p(\bbeta){-}\mathbb{E}_{q(\bbeta)}\log q(\bbeta)] {+}{\sum\nolimits_{i=1}^{n}}\{(y_i{-}0.5){\nabla_{\blambda}}\mathbb{E}_{q(\bbeta)}(\bx_i^{\intercal}\bbeta){-}0.5 \mathbb{E}_{q_{\mbox{\tiny opt}}(z_i)}(z_i){\nabla_{\blambda}}\mathbb{E}_{q(\bbeta)}[(\bx_i^{\intercal}\bbeta)^{2}]\}. && 
\end{eqnarray*} 
To clarify the last equality, note again that $\phi_i(\blambda)=\mathbb{E}_{q(\bbeta)}[-0.5(\bx^{\intercal}_i\bbeta)^2]$, whereas from classical properties of exponential families we also have $\nabla_{\phi_i(\blambda)}\alpha[\phi_i(\blambda)]=\mathbb{E}_{q_{\mbox{\tiny opt}}(z_i)}(z_i)$. This specific form of $\nabla_{\blambda}\mathcal{L}(\blambda)$ provides simple optimization, directly related to \textsc{cavi}. Indeed, comparing the above gradient with the one leading to  equations \eqref{eq14}, it can be noticed that such quantities coincide after replacing each $\mathbb{E}_{q^{(t-1)}(z_i)}(z_i)$ with $\mathbb{E}_{q_{\mbox{\tiny opt}}(z_i)}(z_i)$. Hence, the maximum of $\mathcal{L}(\blambda)$ can be similarly obtained by solving equations \eqref{eq14}, where the expected value is now computed with respect to $q_{\mbox{\tiny opt}}(\bz)$ instead of $q^{(t-1)}(\bz)$.

\begin{figure}[t]
\includegraphics[width=15.3cm]{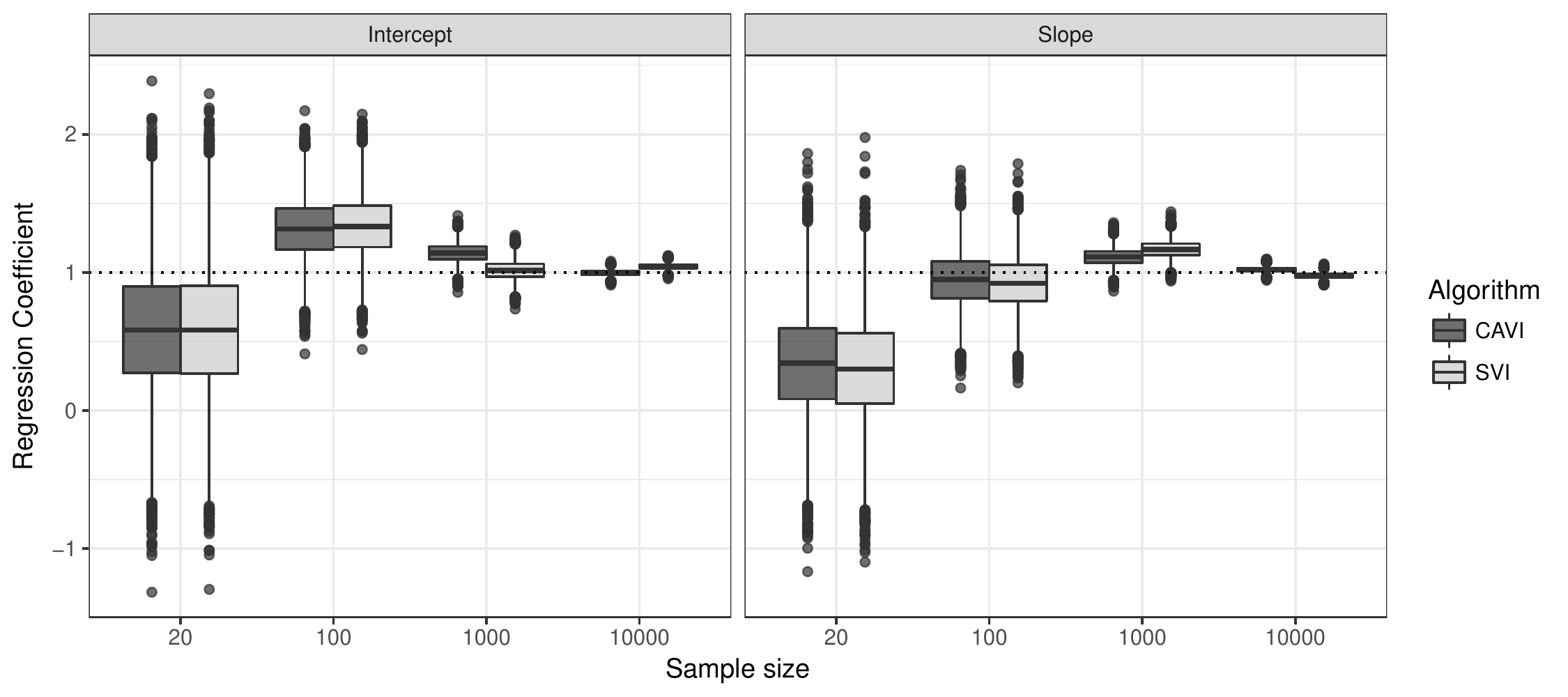}
\caption{{{For increasing $n$, boxplots of the \textsc{cavi} and \textsc{svi} solutions for the coefficients in a Bayesian logistic regression with a single covariate. The dotted horizontal line refers to the true coefficients. }}}
\label{F1}
\end{figure}

To derive the \textsc{svi} routine, let us first express $\mathbb{E}_{q_{\mbox{\tiny opt}}(\bz)}[\boeta_1(\by)]-\blambda_1{=}\ 0$ and $\mathbb{E}_{q_{\mbox{\tiny opt}}(\bz)}[\boeta_2(\bz)]-\blambda_2{=}\ 0$ as
\begin{eqnarray}
\quad\ \ \ \bSigma_0^{-1}\bmu_0+\sum\nolimits_{i=1}^{n}\bx_i(y_i-0.5)-\blambda_1=0, \quad -0.5[\bSigma_0^{-1}+\sum\nolimits_{i=1}^{n}\bx_i\mathbb{E}_{q_{\mbox{\tiny opt}}(z_i)}(z_i)\bx_i^{\intercal}]-\blambda_2=0, 
\label{eq16}
\end{eqnarray}
to highlight how the evaluation of \eqref{eq16} requires storing the entire dataset and summing over the whole units. This step could be a major computational bottleneck when the sample size $n$ is massive, thus motivating optimization  of $\mathcal{L}(\blambda)$   \citep{hoffman_2013} via stochastic approximation of  \eqref{eq16} \citep{robbins_1951}. This is done by constructing a random version of $\{\bSigma_0^{-1}\bmu_0+\sum\nolimits_{i=1}^{n}\bx_i(y_i-0.5)-\blambda_1,-0.5[\bSigma_0^{-1}+\sum\nolimits_{i=1}^{n}\bx_i\mathbb{E}_{q_{\mbox{\tiny opt}}(z_i)}(z_i)\bx_i^{\intercal}]-\blambda_2\}$ whose expectation coincides with these functions, but its realizations are cheaper to compute. A simple solution is to consider the discrete random variable $B(\blambda)$ assuming values  $\{B_i(\blambda_1)=\bSigma_0^{-1}\bmu_0+n\bx_i(y_i-0.5)-\blambda_1,B_i(\blambda_2)=-0.5[\bSigma_0^{-1}+n\bx_i\mathbb{E}_{q_{\mbox{\tiny opt}}(z_i)}(z_i)\bx_i^{\intercal}]-\blambda_2\}$, $i=1, \ldots, n$ with equal probability $n^{-1}$, thus implicitly relying on a mechanism which samples a unit $i$ uniformly and then computes \eqref{eq16} as if such unit was observed $n$ times. This allows application of \citet{robbins_1951} to solve \eqref{eq16} via the iterative updates
\begin{eqnarray}
\blambda_1^{(t)}=\blambda_1^{(t-1)}+\rho_t B_t(\blambda_1^{(t-1)}), \quad \blambda_2^{(t)}=\blambda_2^{(t-1)}+\rho_t B_t(\blambda_2^{(t-1)}), \qquad \mbox{for every iteration $t$},
\label{eq17}
\end{eqnarray}
where $[B_t(\blambda_1^{(t-1)}),B_t(\blambda_2^{(t-1)})]$ is an independent sample of $B(\blambda)$, evaluated at $(\blambda_1^{(t-1)},\blambda_2^{(t-1)})$, whereas $\rho_t$ are step-sizes ensuring convergence to the solution of \eqref{eq16}---and hence to the maximum of $\mathcal{L}(\blambda)$---when $\sum_t \rho_t=+\infty$ and $\sum_t \rho^2_t<+\infty$ \citep{robbins_1951,spall_2005}. \citet{hoffman_2013} set $\rho_t=(t+\tau)^{-\kappa}$, with $\kappa \in (0.5,1]$ denoting the forgetting rate, and $\tau \geq 0$ the delay down-weighting early iterations. These settings ensure the convergence conditions on $\rho_t$. Algorithm \ref{al3} provides the pseudo--code to perform \textsc{svi} in logistic regression under model \eqref{eq10}. As it can be noticed, this routine relies on updating steps which are cheaper to compute than those of \textsc{cavi}. In fact, each iteration of Algorithm  \ref{al3} does not require to sum over the entire dataset, but relies instead on a single observation sampled uniformly. These gains are fundamental to scale--up calculations in massive datasets.

Figure \ref{F1} provides a summarizing quantitative assessment for the performance of \textsc{cavi} and \textsc{svi} in a logistic regression with $\mbox{logit}[\mbox{pr}(y_i=1 \mid \bbeta)]=\beta_1+\beta_2x_i$ for each $i=1, \ldots, n$. To study performance under different dimensions, we generate data for increasing sample size $n \in (20,100,1000,10000)$ from a logistic regression with true coefficients $\beta^{(0)}_1{=} \ \beta^{(0)}_2{=} \ 1$ and covariates $x_1, \ldots, x_n$ from a $\textsc{unif}(-2,2)$. We perform Bayesian inference under a moderately diffuse prior $\bbeta\sim \mbox{N}_2({\bf 0}, 10{\cdot} {\bf I}_2)$ and approximate the posterior via \textsc{cavi} and \textsc{svi}, with $(\tau,\kappa)=(1,0.75)$. As is clear from Figure \ref{F1}, although \textsc{svi} relies on noisy gradients, the final approximations $q_{\textsc{svi}}^*(\beta_1)$, $q_{\textsc{svi}}^*(\beta_2)$ are similar to the optimal solutions $q^*(\beta_1)$, $q^*(\beta_2)$ from  \textsc{cavi}. These approximate posteriors increasingly shrink around the true coefficients as $n$ grows, thus suggesting desirable asymptotic behavior of the \textsc{cavi} and \textsc{svi} solutions. Code and tutorials to reproduce this analysis are available at \url{https://github.com/tommasorigon/logisticVB}.

\section{Discussion}
\label{sec4}
Motivated by the success of the lower bound developed by \citet{jak_2000} for logistic log-likelihoods, and by the lack of formal justifications for its excellent performance, we introduced a novel connection between their construction and a P\'olya-gamma data augmentation developed in recent years for logistic regression \citep{pol_2013}. Besides providing a probabilistic interpretation of the bound derived by \citet{jak_2000}, this connection crucially places the variational methods associated with the proposed lower bound in a more general framework having desirable properties. More specifically, the \textsc{em}  for variational inference proposed by \citet{jak_2000} maximizes a well--defined \textsc{elbo} associated with a conditionally conjugate exponential family model and, hence, provides the same approximation of the \textsc{cavi} for \textsc{vb} in this model.

The above result motivates further generalizations to novel computational methods, including the \textsc{svi} algorithm in Section \ref{svi}. On a similar line of research, an interesting direction is to incorporate the method of \citet{Giordano2015} to correct the variance--covariance matrix in $q^*(\bbeta)$  from Algorithms \ref{al1}--\ref{al3}, which is known to underestimate  variability. Besides this, the results in Figure \ref{F1} motivate also future theoretical studies on the quality of \textsc{cavi} and \textsc{svi} approximations in asymptotic settings. This can be done by adapting available theory on mean--field \textsc{vb}  for conditionally conjugate exponential family models \citep[e.g.][]{wang2004}. Finally, we shall also emphasize that although our focus is on classical Bayesian logistic regression, the results in Sections \ref{sec2}--\ref{sec3} can be easily generalized to more complex learning procedures incorporating logistic models as a building--block, as long as such formulations admit conditionally conjugate exponential family representations.

\section*{Appendix A. Maximum likelihood estimation}
\label{appendix}
Although maximum likelihood estimation of the parameters in logistic regression is well--established, there is still active research within this class of models to address other important open questions. For instance, classical Newton--Raphson does not guarantee monotone log-likelihood sequences, thus potentially affecting  the stability of the maximization routine \citep{boh_1988}.  This has motivated other methods leveraging alternative quadratic approximations which uniformly minorize the logistic log-likelihood and are tangent to it \citep{boh_1988,de_2009,brow_2015}, thus guaranteeing  monotone convergence \citep{hun_2004}. As discussed in Sections \ref{sec1}--\ref{sec2}, this is the case of the bound \eqref{eq5} in \citet{jak_2000}.

Motivated by this result, \citet{jak_2000} provided in Appendix C of their article an iterative routine for maximum likelihood estimation of $\bbeta$ that has monotone log-likelihood sequences and simple maximizations. In particular, letting $\bbeta^{(t-1)}$ denote the estimate of the coefficients at step $t-1$ and simplifying the calculations in Appendix C of  \citet{jak_2000}, their routine first maximizes \eqref{eq5} with respect to $\xi_1, \ldots, \xi_n$, obtaining $\hat{\xi}^{(t-1)}_i=\bx_i^{\T}\bbeta^{(t-1)}$, $i=1,\ldots, n$, and then derive $\bbeta^{(t)}$ by maximizing $\log {\bar{p}}^{(t-1)}(\by\mid \bbeta)=\sum_{i=1}^n\log \bar{p}^{(t-1)}(y_i\mid \bbeta)$, with $\xi_i$ replaced by $\hat{\xi}^{(t-1)}_i=\bx_i^{\T}\bbeta^{(t-1)}$. This last optimization is straightforward to compute due to the quadratic form of \eqref{eq5}, thus providing 
\begin{eqnarray}
\bbeta^{(t)}=(\bX^{\T} \hat{\bZ}^{(t-1)} \bX)^{-1}\bX^{\T}(y_1-0.5, \ldots, y_n-0.5)^{\T}=(\bX^{\T} \hat{\bZ}^{(t-1)} \bX)^{-1}\bX^{\T}(\by-0.5{\cdot}\mbox{\bf 1}_n),
\label{eq3append}
\end{eqnarray}
with $\hat{\bZ}^{(t-1)}$ a diagonal matrix having entries $\hat{\bZ}^{(t-1)}_{[ii]}=0.5(\bx_i^{\T}\bbeta^{(t-1)})^{-1}\mbox{tanh}[0.5(\bx_i^{\T}\bbeta^{(t-1)})]$. Analyzing this strategy in the light of \eqref{eq8}, it can be noticed that $\hat{\xi}^{(t-1)}_i=\bx_i^{\T}\bbeta^{(t-1)}$ leads to the solution $\hat{q}^{(t-1)}(z_i)$ minimizing the \textsc{kl} divergence $\textsc{kl}[q(z_i)  \mid\mid  p(z_i \mid \bbeta^{(t-1)})]$ in \eqref{eq8}, for $i=1, \ldots, n$, whereas the function 
\begin{eqnarray*}
\log {\bar{p}}^{(t-1)}(\by\mid \bbeta)=\sum_{i=1}^n\log  {\bar{p}}^{(t-1)}(y_i\mid \bbeta)=\sum_{i=1}^n\int_{\Re_{+}}\hat{q}^{(t-1)}(z_i) \log \frac{p(y_i\mid \bbeta)p(z_i\mid \bbeta)}{\hat{q}^{(t-1)}(z_i)}\mbox{d}z_i,
\end{eqnarray*}
maximized with respect to $\bbeta$ is equal, up to an additive constant, to the expectation $Q(\bbeta \mid \bbeta^{(t-1)})$ of the complete log-likelihood $\log p(\by,\bz \mid \bbeta)=\sum_{i=1}^n\log[p(y_i\mid \bbeta)p(z_i \mid \bbeta)]$ computed with respect to the conditional distribution  of the P\'olya-gamma data $(z_i \mid \bbeta^{(t-1)}) \sim \textsc{pg}(1,\bx_i^{\T}\bbeta^{(t-1)})$, for $i=1, \ldots, n$. Combining these results with the key rationale underlying  \textsc{em}  \citep[e.g.][Ch. 9.4]{bishop_2006}, it  follows that the routine in Appendix C of  \citet{jak_2000}   is an \textsc{em} based on P\'olya-gamma augmented data. This algorithm first computes the expectation $Q(\bbeta \mid \bbeta^{(t-1)})=\sum_{i=1}^n\mathbb{E}_{\hat{q}^{(t-1)}(z_i)}\log[p(y_i\mid \bbeta)p(z_i \mid \bbeta)]$ and then maximizes it with respect to $\bbeta$. See also  \citet{scott_2013}.

As already discussed in \citet{jak_2000}, the above maximization method guarantees a monotone log-likelihood sequence, thus ensuring a stable convergence. Indeed, leveraging \eqref{eq5} and  \eqref{eq8}, it can be noticed that the above routine is a genuine minorize-majorize (\textsc{mm}) algorithm \citep[e.g.][]{hun_2004}, provided that $\sum_{i=1}^n\log p(y_i \mid \bbeta) \geq \sum_{i=1}^n\log  {\bar{p}}^{(t-1)}(y_i\mid \bbeta)$ for each $\bbeta$, and  that $\sum_{i=1}^n\log p(y_i \mid \bbeta^{(t-1)}) = \sum_{i=1}^n\log  {\bar{p}}^{(t-1)}(y_i\mid \bbeta^{(t-1)})$. We shall notice that also \citet{de_2009} and \citet{brow_2015}  highlighted this relation with \textsc{mm} under a mathematical argument, while discussing the sharpness of  \eqref{eq5}. Exploiting results in Section~\ref{sec2}, we also show that the \textsc{mm} algorithm relying on \eqref{eq5} improves the convergence rate of the one in \citet{boh_1988}. To our knowledge, this is the only tractable \textsc{mm} alternative to \citet{jak_2000}.

To address the above goal, let us first re--write \eqref{eq3append} to allow a direct comparison with the solution 
\begin{eqnarray}
\bbeta^{(t)}=\bbeta^{(t-1)}+(\bX^{\T}\bGamma \bX)^{-1}\bX^{\T}(y_1-\pi_1^{(t-1)}, \ldots, y_n-\pi_n^{(t-1)})^{\T},
\label{eq4append}
\end{eqnarray}
from \citet{boh_1988}, with $\bGamma=0.25{\cdot}\mbox{\bf I}_n$ and $\pi_i^{(t-1)}=[1+\exp (-{\bx}_i^{\T} {\bbeta}^{(t-1)})]^{-1}$. Indeed, by adding and  subtracting $\hat{\bZ}^{(t-1)}\bX\bbeta^{(t-1)}$ inside $(\by-0.5{\cdot}\mbox{\bf 1}_n)$, equation \eqref{eq3append} coincides with 
\begin{eqnarray}
\bbeta^{(t)}=\bbeta^{(t-1)}+(\bX^{\T} \hat{\bZ}^{(t-1)} \bX)^{-1}\bX^{\T}(y_1-\pi_1^{(t-1)}, \ldots, y_n-\pi_n^{(t-1)})^{\T}, 
\label{eq5append}
\end{eqnarray}
after noticing that each element $y_i- 0.5-0.5\mbox{tanh}[0.5({\bx}_i^{\T}{\bbeta}^{(t-1)})]$ in  $\by-0.5{\cdot}\mbox{\bf 1}_n-\hat{\bZ}^{(t-1)}\bX\bbeta^{(t-1)}$ can be alternatively expressed as $y_i-0.5\{1+[1-\exp (-{\bx}_i^{\T} {\bbeta}^{(t-1)})][1+\exp (-{\bx}_i^{\T}{\bbeta}^{(t-1)})]^{-1}\}=y_i-\pi_i^{(t-1)}$. A closer inspection of  \eqref{eq4append} and \eqref{eq5append} shows that the updating underlying \citet{boh_1988} and \citet{jak_2000} coincides with the one from the Newton--Raphson, after replacing the Hessian $\bH^{(t-1)}=-\bX^{\T}\bLambda^{(t-1)}\bX=-\bX^{\T}\mbox{diag}[\pi^{(t-1)}_1(1-\pi^{(t-1)}_1), \ldots, \pi^{(t-1)}_n(1-\pi^{(t-1)}_n)]\bX$---evaluated in $\bbeta^{(t-1)}$---of the logistic log-likelihood, with $-\bX^{\T}\bGamma \bX$ in \eqref{eq4append}  and $-\bX^{\T} \hat{\bZ}^{(t-1)} \bX$ in  \eqref{eq5append}. Recalling \citet{boh_1988}, both matrices define a lower bound for the Hessian and guarantee that both \eqref{eq4append} and \eqref{eq5append} induce a monotone  sequence for the log-likelihood. In \citet{boh_1988}  the uniform bound follows after noticing that $\bpi^{(t-1)}(1-\bpi^{(t-1)})\leq 0.25{\cdot}\mbox{\bf 1}_n$ for any $\bpi^{(t-1)} \in (0,1)^n$, whereas, according to Lemma \ref{lemma1}, the adaptive bound induced by  \citet{jak_2000} is formally related to an exact data augmentation, thus suggesting that \eqref{eq5append} may provide  more efficient updates than  \eqref{eq4append}. This claim is formalized in Proposition~\ref{theorem2} by comparing the convergence rates of the two algorithms. Refer to \citet[][Chapter 3.9]{mcla_2007} for details regarding the definition and the computation of the convergence rate associated with a generic iterative routine.

\begin{proposition}
\label{theorem2}
Assume that $\bbeta^*$ is the limit, if it exists, of $\{\bbeta^{(t)}: t \geq 1\}$, and let $\mbox{M}^{\textsc{b}}(\cdot)=\{\mbox{M}_1^{\textsc{b}}(\cdot), \ldots, \mbox{M}_p^{\textsc{b}}(\cdot) \}$ and $\mbox{M}^{\textsc{j}}(\cdot)=\{\mbox{M}_1^{\textsc{j}}(\cdot), \ldots, \mbox{M}_p^{\textsc{j}}(\cdot) \}$ denote the functions mapping from $\bbeta^{(t-1)}$ to $\bbeta^{(t)}$ in \eqref{eq4append} and \eqref{eq5append}, respectively. Then $r_{\textsc{b}}\geq r_{\textsc{j}}$, with $r_{\textsc{b}}= ||  \bJ_{\textsc{b}}^*  ||_2$ and $r_{\textsc{j}}= ||  \bJ_{\textsc{j}}^*  ||_2$ being the maximum eigenvalues of the Jacobians $\bJ_{\textsc{b}}=\partial \mbox{M}^{\textsc{b}}(\bbeta)/\partial \bbeta$ and $\bJ_{\textsc{j}}=\partial \mbox{M}^{\textsc{j}}(\bbeta)/\partial \bbeta$, respectively, computed in $\bbeta^*$.
\end{proposition}

To prove Proposition~\ref{theorem2}, note that $\bJ_{\textsc{b}}^*={\bf I}_p+(\bX^{\T}\bGamma \bX)^{-1} \bH^{*}$ can be easily computed as in  \citet{boh_1988}, since $\bGamma$ does not depend on $\bbeta$ in \eqref{eq4append}. It is instead not immediate to calculate $\bJ_{\textsc{j}}^*$ via direct differentiation of $\mbox{M}^{\textsc{j}}(\bbeta)$, because \eqref{eq5append} contains more complex hyperbolic transformations of $\bbeta$. However, exploiting the probabilistic findings in Section~\ref{sec2}, this issue can be easily circumvented  leveraging the \textsc{em} interpretation of the routine in \citet{jak_2000} via P\'olya-gamma augmented data. Indeed, following  \citet[][Ch. 3.9.3]{mcla_2007},  the rate matrix of an iterative routine based on \textsc{em} methods, is  equal to $\bJ_{\textsc{j}}^*=\mbox{\bf I}_p+\mathcal{ I}_c(\bbeta^*)^{-1}  \bH^{*}$, with $\mathcal{ I}_c(\bbeta^*)$ denoting the expectation, taken with respect to the  augmented data, of the complete-data information matrix $\mbox{I}_c(\bbeta^*)$. This quantity can be easily computed in our case, provided that the complete log-likelihood  is equal, up to an additive constant, to the quadratic function $\sum_{i=1}^n\{(y_i-0.5)\bx_i^{\T}\bbeta-0{.5}z_i(\bx_i^{\T}\bbeta)^2\}$ of $\bbeta$, which is also linear in the augmented  P\'olya-gamma data $z_i$. Due to this, it is easy to show that $\mathcal{ I}_c(\bbeta^*)=\bX^{\T}\hat{\bZ}^*\bX$, where $\hat{\bZ}^*$ is the $n \times n$ diagonal matrix with entries $\hat{\bZ}_{[ii]}^*=0.5(\bx_i^{\T}\bbeta^{*})^{-1}\mbox{tanh}\{0{.}5(\bx_i^{\T}\bbeta^{*})\}$.

\begin{proof}[Proof]
Recalling the above discussion, the proof of Proposition~\ref{theorem2} requires comparing the maximum eigenvalues of $\bJ_{\textsc{b}}^*=\mbox{\bf I}_p+(\bX^{\T}\bGamma \bX)^{-1} \bH^{*}=\mbox{\bf I}_p-(\bX^{\T}\bGamma \bX)^{-1}(\bX^{\T}\bLambda^{*}\bX)$ and $\bJ_{\textsc{j}}^*=\mbox{\bf I}_p+(\bX^{\T}\hat{\bZ}^*\bX)^{-1} \bH^{*}=\mbox{\bf I}_p-(\bX^{\T}\hat{\bZ}^* \bX)^{-1}(\bX^{\T}\bLambda^{*}\bX)$. To address this goal, let us first show that $\bJ_{\textsc{b}}^* \geq \bJ_{\textsc{j}}^*$. This result can be proved by noticing that $\bJ_{\textsc{b}}^*$ and $\bJ_{\textsc{j}}^*$ only differ by the positive diagonal matrices $\bGamma$ and $\hat{\bZ}^*$, respectively. Hence, the above inequality is met if all the entries $ \bGamma_{[ii]} - \hat{\bZ}^*_{[ii]}=0.25-0.5(\bx_i^{\T}\bbeta^{*})^{-1} \mbox{tanh}\{0.5(\bx_i^{\T}\bbeta^{*})\}=0.25[1-2(\bx_i^{\T}\bbeta^{*})^{-1}\mbox{tanh}\{0{.}5(\bx_i^{\T}\bbeta^{*})\}]$ in the diagonal matrix $\bGamma-\hat{\bZ}^*$ are non-negative.  Letting $u=0.5(\bx_i^{\T}\bbeta^{*})$, and re--writing the well--established inequality $u^{-1}\mbox{sinh}(u) \leq \mbox{cosh}(u)$ \citep[e.g.][]{zhu_2012} as $u^{-1}\mbox{sinh}(u)\mbox{cosh}^{-1}(u)\leq 1$, it follows that $u^{-1}\mbox{tanh}(u)=2(\bx_i^{\T}\bbeta^{*})^{-1} \mbox{tanh}\{0.5(\bx_i^{\T}\bbeta^{*})\}\leq 1$, thus guaranteeing $\bGamma-\hat{\bZ}^* \geq 0$, and, as a direct consequence, $\bJ_{\textsc{b}}^* \geq \bJ_{\textsc{j}}^*$. This concludes the proof. In fact, $\bJ_{\textsc{b}}^* \geq \bJ_{\textsc{j}}^*$ implies $ ||  \bJ_{\textsc{b}}^*  ||_2=r_{\textsc{b}} \geq r_{\textsc{j}}= ||  \bJ_{\textsc{j}}^*  ||_2$ \citep[e.g.][]{knut_2001}.
\end{proof}

Proposition~\ref{theorem2} ensures that \eqref{eq5append} improves the convergence rate of  \eqref{eq4append}. In fact, higher values of $r$ imply slower convergence. We shall however emphasize that the \textsc{em} in Appendix C of  \citet{jak_2000} does not reach the quadratic convergence of Newton--Raphson, but guarantees monotone log-likelihood sequences. It is also important to highlight that although the \textsc{mm} in  \citet{boh_1988} has slower convergence, the matrix $\bX^{\T}\bGamma \bX$ in \eqref{eq4append} does not depend on $\bbeta^{(t-1)}$, thus requiring inversion only once during the iterative procedure. This reduces computational complexity, especially in high--dimensional problems, compared to the updating in  \eqref{eq5append}, which requires, instead, inversion of $\bX^{\T} \hat{\bZ}^{(t-1)} \bX$ at each iteration. We refer to the tutorial \texttt{em\_logistic\_tutorial.md} in \url{https://github.com/tommasorigon/logisticVB} for illustrative simulations providing a quantitative comparison among the aforementioned methods.

Although the above focus has been on maximum likelihood estimation methods, the probabilistic interpretation \eqref{eq7} of the quadratic bound in  \citet{jak_2000} motivates simple adaptations to include the maximum a posteriori estimation problem under a Bayesian framework. This routine has been carefully studied by  \citet{scott_2013} and we refer to their contribution for details.

\bibliographystyle{imsart-nameyear}
\bibliography{paper-ref}

\end{document}